\documentclass[a4paper,11pt]{article}

\usepackage{graphicx}% http://ctan.org/pkg/graphicx

\usepackage{fullpage}
\usepackage{libertine}
\usepackage{color}

\usepackage{mdframed}

\usepackage[ruled,vlined]{algorithm2e}
\SetArgSty{textrm}

\usepackage{amsmath,amsfonts,amssymb,amsthm}

\usepackage[breaklinks=true]{hyperref}
\usepackage[svgnames]{xcolor}
\usepackage[capitalise,nameinlink]{cleveref}
\hypersetup{colorlinks={true},linkcolor={DarkBlue},citecolor=[named]{DarkGreen}}

\usepackage{natbib}%\usepackage[numbers]{natbib}

\usepackage{subcaption}

\usepackage{tikz}  %needed for \textcolor as well as tikz
\usetikzlibrary{arrows}
\usetikzlibrary{patterns,snakes}
\usetikzlibrary{decorations.shapes}
\tikzstyle{overbrace text style}=[font=\tiny, above, pos=.5, yshift=5pt]
\tikzstyle{overbrace style}=[decorate,decoration={brace,raise=5pt,amplitude=3pt}]
\usetikzlibrary{shapes.geometric}

\newtheorem{theorem}{Theorem}[section]

\newtheorem{lemma}[theorem]{Lemma}

\theoremstyle{definition}
\newtheorem*{comment*}{Comment}

\newtheorem{remark}[theorem]{Remark}

\newcommand{\SW}{\text{\normalfont SW}}

\newcommand{\bv}{\mathbf{v}}

\newcommand{\info}{\text{info}}
\newcommand{\consistency}{\text{consistency}}
\newcommand{\plu}{\text{\normalfont plu}}
\newcommand{\robustness}{\text{robustness}}
\newcommand{\distortion}{\text{distortion}}

\renewcommand{\top}{\text{\normalfont top}}

\newcommand{\mech}{\mathcal{M}}

\newcommand{\reliable}{accurate\xspace}
\newcommand{\unreliable}{unreliable\xspace}

\newcommand{\ceil}[1]{\lceil {#1} \rceil}
\newcommand{\Mod}[1]{\ (\mathrm{mod}\ #1)}

\title{\bf Utilitarian Distortion with Predictions}

\usepackage{authblk}
\author[1]{Aris Filos-Ratsikas}
\author[1]{Georgios Kalantzis}
\author[2]{Alexandros A. Voudouris}

\affil[1]{University of Edinburgh, UK}
\affil[2]{University of Essex, UK}

\date{}

\date{}

\begin{document}

\allowdisplaybreaks

\maketitle

\begin{abstract}
We study the utilitarian distortion of social choice mechanisms under the recently proposed learning-augmented framework where some (possibly unreliable) \emph{predicted information} about the preferences of the agents is given as input. In particular, we consider two fundamental social choice problems: \emph{single-winner voting} and \emph{one-sided matching}. In these settings, the ordinal preferences of the agents over the alternatives (either candidates or items) is known, and some prediction about their underlying cardinal values is also provided. The goal is to leverage the prediction to achieve improved distortion guarantees when it is accurate, while simultaneously still achieving reasonable worst-case bounds when it is not. This leads to the notions of \emph{consistency} and \emph{robustness}, and the quest to achieve the best possible tradeoffs between the two. We show \emph{tight} tradeoffs between the consistency and robustness of ordinal mechanisms for single-winner voting and one-sided matching, for different levels of information provided by as prediction. 
\end{abstract}

%%%%%%%%%%%%%%%%%%%%%%%%
%%%%%%%%%%%%%%%%%%%%%%%%
\section{Introduction}
The objective of social choice theory is to aggregate the preferences of individuals into a collective outcome which is aligned with the well-being of society. In this quest, a fundamental challenge is the accurate elicitation of those preferences, which might be constrained by several factors, such as cognitive or computational. The literature of \emph{distortion} in social choice theory \citep{procaccia2006distortion} over the past 20 years has been concerned with precisely this question, i.e., the effect that the limits on the amount of available information can have on the quality of the outcomes. Concretely, in the original version of the problem (as studied by \cite{boutilier2015optimal}), the intrinsic preferences of the participants, or \emph{agents}, are captured by \emph{cardinal valuation functions}\footnote{These can be interpreted as von Neumann-Morgenstern utilities \citep{von1944theory}.}, but their expressiveness is limited to \emph{ordinal preference rankings} consistent with these values. This is motivated by the fact that sorting different options is usually cognitively a much easier task compared to coming up with appropriate numerical values to assign to them. For a mechanism that operates under this ordinal information regime, the distortion measures (typically in a worst-case, multiplicative sense), how far the outcome computed by the mechanism is from the ideal outcome that would be possible with access to all the (cardinal) values. 

A recent line of work has interpreted distortion in a broader sense, where the limited information is not restricted to be of ordinal nature (see \citep[Section 5]{distortion-survey}). Specifically, these works assume that, in addition to the preference rankings, the mechanism has also access to some cardinal information about the values of the agents. The goal then is to achieve the best possible tradeoffs between the amount of cardinal information and the achievable distortion. This information is typically accessed via different types of queries \citep{ABFV21,amanatidis2022few,ABFV24,ma2021matching,caragiannis2023beyond,ES25} or the elicitation of threshold approval preferences \citep{bhaskar2018truthful,abramowitz2019awareness,benade2021preference,latifian2024approval,anshelevich2025improved}. Regardless of the method, the commonly presented motivation for this approach is that {\em eliciting a small amount of cardinal information is not very demanding, while it could potentially have wondrous effects on the social welfare}~\citep{amanatidis2022few}.

While certainly reasonable, this approach is based on the premise that agents will be able to express accurately their  values, as long as they are asked to do so for only a few outcomes. This is contingent on assumptions about the cognitive burden that the elicitation process imposes, which, in reality, highly depends on the application at hand. Ideally, we do not only want our mechanisms to take advantage of this cardinal information, but also be robust against possible errors in the agents' judgment, which could potentially render the values they report much different than their true ones. In simple words, we would like to design mechanisms with improved distortion bounds when some accurate cardinal information is provided, while at the same time do not (\emph{significantly}, or even \emph{at all}) compromise their worst-case distortion guarantees when this information is inaccurate. 

This very principle is the backbone of the literature on \emph{learning-augmented} algorithm design \citep{lykouris2021competitive}. In this framework, an algorithm, equipped with a (potentially unreliable) \emph{prediction} about some input parameters of the problem, aims to achieve improved levels of performance (\emph{consistency}) when this prediction is correct, while still achieving good worst-case bounds when it is not (\emph{robustness}). In the original motivation for this framework, the prediction was assumed to be provided by an appropriate machine learning model. However, the setup is actually much more general; indeed, the prediction can come from any source of potentially unreliable information. The elicitation of cardinal values in social choice settings is an inherently error-prone process, and as such, it is quite meaningful to interpret the outcomes of this elicitation as predictions. This brings us to the following central question studied in our work:
\begin{quote}
In social choice settings, where the input consists of the ordinal preferences of the agents and a \emph{prediction} about their cardinal values, what are the best possible tradeoffs between consistency and robustness?
\end{quote}

\subsection{Our Contribution}
We study the question above for two fundamental \emph{utilitarian} social choice settings, that of \emph{single winner voting} and \emph{one-sided matching}, when agents have unit-sum values, and the goal is to maximize the \emph{social welfare} (the total value of the agents). For both settings, we design mechanisms that achieve \emph{tight} tradeoffs between consistency and robustness, by also showing that no other mechanism can achieve better tradeoffs with the given prediction. 

In the the single-winner voting setting, where a set of $n$ voters aim to select one out of $m$ candidates, the distortion of the best purely ordinal deterministic mechanisms is known to be $\Theta(m^2)$ \citep{caragiannis2011voting,caragiannis2017subset}. We first show a stark impossibility: Any mechanism that has bounded robustness cannot achieve asymptotically sublinear distortion. This refines the search for good mechanisms to those that achieve linear distortion and robustness as close as possible to the aforementioned $\Theta(m^2)$ bound. To this end, we show that the best possible tradeoff is a consistency of $\Theta(\lambda m)$ and a robustness of $\Theta(m^3/\lambda)$, by providing a mechanism that achieves it, as well as a general asymptotically matching lower bound that applies to all mechanisms for this setting. It is worth noting that our mechanism only uses predicted information about the most-preferred candidate of each voter, whereas all of our lower bounds apply to the case where the whole valuation profile is provided as prediction; this makes our tradeoff bounds as strong as possible. We conclude our investigation for the single-winner voting setting with a refined analysis of our mechanism which bounds its distortion as a function of an appropriately defined \emph{prediction error}.

We next consider the one-sided matching setting, in which a set of $n$ agents express preferences over a set of $n$ items, and the goal is to come up with a matching of agents to items with as high social welfare as possible. In this setting, \citet{amanatidis2022few} proved that the best possible purely ordinal deterministic mechanism achieves a distortion of $\Theta(n^2)$. In the realm of mechanisms with predictions, we first observe that, if the prediction includes the whole valuation profile, then it is straightforward to achieve a consistency of $1$ while still matching the $\Theta(n^2)$ bound in terms of robustness. For this reason, we consider the case where the prediction is truncated to the values of the agents for their $k$ most-preferred items, and propose a mechanism which achieves consistency $O(n/k)$ and robustness $O(n^2)$; this is proved to be tight, i.e., no other mechanism can achieve asymptototically better consistency and robustness bounds. As in the case of single-winner voting, we provide a fine-grained analysis of the distortion of our mechanism in terms of a prediction error. 

Our results can be interpreted as follows. For voting, it is impossible to achieve consistency $o(m^2)$ while guaranteeing an ideal bound of $O(m^2)$ on the robustness. One can, however, achieve any possible consistency-robustness tradeoff between ($O(m^2),O(m^2)$) and ($O(m), O(m^3)$) by choosing different values for the parameter $\lambda$, depending on how much the prediction can be trusted. For one-sided matching, the main takeaway message is that improved distortion guarantees can be achieved when the prediction is accurate (where the improvement depends linearly on the amount of available predicted information, captured by the value of $k$), without compromising the $O(n^2)$ distortion bound when the prediction is inaccurate. 

Finally, we remark that our focus on \emph{deterministic} mechanisms (which select a single outcome rather than a probability distribution over outcomes) is because there are straighforward randomized mechanisms that trivially achieve the best possible consistency-robustness tradeoff in both settings. In particular, we can equiprobably select among the welfare-maximizing outcome (according to the given prediction) and the outcome of the best possible purely ordinal mechanism from the literature. If the prediction is accurate, the achieved social welfare is at least half that of the maximum possible, and if it is inaccurate, the social welfare is at least half of that of the best ordinal mechanism, which is within an optimal factor of $O(m^2)$ (for social choice) or of $O(n^2)$ (for one-sided matching) from the best possible. Since we are interested in asymptotic bounds, such an investigation is not meaningful. 

\subsection{Related Work and Discussion}
The single-winner voting setting is almost synonymous with the general social choice setting, as it can be used to model decision-making scenarios where a set of individuals collectively select an outcome; see \citep{brandt2012computational}. In the distortion literature, it has been the main setting of study since the inception of the framework, giving rise to a plethora of papers in both the utilitarian setting \citep{caragiannis2011voting,caragiannis2017subset,filos2014truthful,boutilier2015optimal} and the metric setting, where the cardinal information is captured by metric costs \citep{anshelevich2017randomized,anshelevich2018approximating,caragiannis2022multiwinner,charikar24breaking,feldman2016voting,gkatzelis2020resolving,kempe2022veto}. The most important results in this literature for us are the worst-case distortion bounds of $O(m^2)$ and $\Omega(m^2)$, proven by \citet{caragiannis2011voting} and \citet{caragiannis2017subset}, respectively. We refer to the survey of \citet{distortion-survey} for more details about other works. 

The one-sided matching setting was introduced (as a social choice problem) in the pioneering work of \citet{hylland1979efficient}, and was studied extensively ever since in economics and computer science; e.g., see \citep{abdulkadirouglu1998random,abdulkadirouglu2005boston,bogomolnaia2001new,Anshelevich2016truthful,Anshelevich2019bipartitematching,caragiannis2024augmentation}. In the literature of utilitarian distortion, it was notably studied by \citet{filos2014RSD}, who considered randomized mechanisms, and then by \cite{amanatidis2022few,ABFV24,ma2021matching} and \citet{ES25} in settings where \emph{reliable} cardinal information is also provided to the mechanism via queries. We remark that \cite{ABFV24} and \cite{ES25} in fact studied both single-winner voting and one-sided matching, as we do here; in a certain sense, these two settings can be considered as the canonical ones for utilitarian distortion. 

Besides the literature of distortion, our work obviously has connections to the literature on learning-augmented algorithms. As we mentioned before, the associated research agenda was explicitly put forward by \citet{lykouris2021competitive} to study problems where the prediction comes as advice for an appropriate machine learning model. The associated literature has revisited several classic algorithmic problems under this lens, including problems in data structures, online algorithms, approximation algorithms, combinatorial optimization, and sublinear algorithms, among others. We refer the reader to the survey of \citet{mitzenmacher2022algorithms} for a more in depth discussion of this literature, as well as the ``Algorithms with Predictions'' online repository \citep{ALPS}. 

In more economically oriented applications, the concept of predictions has been applied to many different problems related to mechanism design \citep{agrawal2022learning,BGT23,balkanski2024randomized,balkanski2023online, BGI24, christodoulou2024mechanism, QNS24}, bidding auctions \citep{XL22,lu2024competitive,caragiannis2024randomized,gkatzelis2025clock}, strategyproof mechanisms for learning and assignment problems \citep{BZ25,colini2024trust}, and cost sharing \citep{gkatzelis2022improved}. Most closely related to us is the work of \cite{BFGT24}, who studied of distortion of metric single-winner voting with predictions and showed tight consistency and robustness tradeoffs for that setting. The metric distortion setting is markedly different from the utilitarian setting that we study, and therefore their results do not have any implications in our setting. Still, it is interesting to note that in their case, the best possible tradeoff is achievable by mechanisms that only require information about the identity of the optimal candidate. In contrast, we show that such a prediction is insufficient to achieve any improvements over purely ordinal mechanisms. Still, similarly to \citet{BFGT24}, our mechanism that achieves the best possible tradeoff does not require prediction information about the whole cardinal profile, but only about the value of each voter for her most-preferred candidate. 
%%%%%%%%%%%%%%%%%%%%%%%%
%%%%%%%%%%%%%%%%%%%%%%%%

%%%%%%%%%%%%%%%%%%%%%%%%
%%%%%%%%%%%%%%%%%%%%%%%%
\section{Preliminaries}
%%%%%%%%%%%%%%%%%%%%%%%%
%%%%%%%%%%%%%%%%%%%%%%%%
We consider two fundamental social choice settings, {\em single-winner voting} and {\em one-sided matching}. In the voting setting, there is a set $N$ of $n$ {\em voters} and a set $A$ of $m$ {\em candidates}. Each voter $i \in N$ has a {\em valuation function} $v_i: A \rightarrow \mathbb{R}_{\geq 0}$ that maps each candidate $x \in A$ to a real non-negative value $v_i(x)$ representing the happiness level of $i$ if $x$ is chosen. As typical in the literature, we assume that the valuation functions satisfy the {\em unit-sum normalization} which requires that $\sum_{x \in A} v_i(x) = 1$ for each voter $i$. We denote by $\bv = (v_i)_{i \in N}$ the {\em valuation profile} of all voters. The {\em social welfare} of a candidate $x \in A$ is defined as the total value of the voters for $x$, that is, 
\begin{align*}
    \SW(x|\bv) = \sum_{i \in N} v_i(x). 
\end{align*}

In the one-sided matching setting, the setup is very similar. There is a set $N$ of {\em agents} and a set $M$ of {\em items} with $|N| = |M| = n$. Each agent has a (unit-sum normalized) valuation function $v_i:M \rightarrow \mathbb{R}$ assigning non-negative values to the different items. An outcome is a one-to-one {\em matching} $\mu$ between the agents and the items; we denote by $\mu(i)$ the item that is assigned to agent $i$, and by $\mu(x)$ the agent that is assigned to item $x$. The social welfare of a matching $\mu$ is 
\begin{align*}
    \SW(\mu|\bv) = \sum_{i \in N} v_i(\mu(i)). 
\end{align*}

%%%%%%%%
\subsection{Mechanisms and Distortion}
%%%%%%%%
A {\em mechanism} $\mech$ takes as input some information $\info(\bv)$ about the valuation profile and outputs an outcome $\mech(\info(\bv))$ that is either a candidate (in the voting setting) or a matching (in the one-sided matching setting). Typically, the available information consists of the {\em ordinal} preferences of the voters which are induced by their valuation functions. In particular, each voter $i$ reports an {\em ordering} (or, {\em ranking}) $\succ^\bv_i$ of the candidates or the items such that $x \succ^\bv_i y$ implies $v_i(x) \geq v_i(y)$. Let $\succ_\bv = (\succ^\bv_i)_{i \in N}$ be the {\em ordinal profile} of the voters induced by $\bv$.

The {\em distortion} is a measure of the loss in social welfare due to making decision based on incomplete information about the valuation functions of the agents. In particular, the distortion of a mechanism that takes as input only the ordinal profile $\succ_\bv$ is defined as the worst-case (over all possible instances) of the maximum social welfare achieved by any outcome over the social welfare achieved by the outcome computed by the mechanism, that is, 
\begin{align*}
    \distortion(\mech) = \sup_{\bv} \frac{\max_{z} \SW(z|\bv) }{ \SW( \mech(\succ_\bv) | \bv) }.
\end{align*}

%%%%%%%%
\subsection{Mechanisms with Predictions}
%%%%%%%%
In the learning-augmented framework, we consider classes of mechanisms that make decisions not only based on the ordinal profile $\succ_\bv$ of the voters, but also on some given {\em prediction information} $p$. While $p$ might be rather general, we will focus on cases where it is an outcome with desirable properties (such as the one maximizing the social welfare), or some (partial) valuation profile $\hat{\bv}$ which also induces the same ordinal profile as the true unknown valuation profile $\bv$ (that is, $\succ_{\hat{\bv}} = \succ_\bv$). 
This predicted information   
might be close to the actual true hidden information, or not. 
We will write $\bv \rhd p$ to denote that $p$ is \emph{accurate} for the true valuation profile. In particular, $\bv \rhd p$ implies that:
\begin{itemize}
    \item If $p$ is the optimal outcome $z^*$, then $z^* \in \arg\max_z \SW(z|\bv)$.
    \item If $p$ is a (partial) valuation profile $\hat{\bv}$, then $\hat{v}_i(x) = v_i(x)$ for each pair of agent $i$ and candidate or item $x$ explicitly defined in $\hat{\bv}$. 
\end{itemize}
We will denote by $\mech(\succ_\bv,p)$ the outcome computed by a mechanism $\mech$ that takes as input the ordinal profile $\succ_\bv$ and some predicted information $p$.

Since the decisions made by our mechanisms rely on some predicted information, we are interested in quantifying the distortion depending on whether the prediction is accurate or not. The {\em consistency} of a mechanism $\mech$ is the worst-case ratio (over all valuation profiles for which the prediction is accurate) between the maximum possible social welfare of any outcome and the social welfare of the outcome chosen by $M$:
\begin{align*}
    \consistency(\mech) = \sup_{\bv: \bv \rhd p} \frac{\max_z \SW(z|\bv) }{ \SW( \mech(\succ_\bv, p) | \bv) }.
\end{align*}
The {\em robustness} of a mechanism $M$ is the same ratio but with the worst-case taken over all valuation profiles and predictions:
\begin{align*}
    \robustness(\mech) = \sup_{\bv, p} \frac{\max_z \SW(z|\bv) }{ \SW( \mech(\succ_\bv, p) | \bv) }.
\end{align*}

%%%%%%%%%%%%%%%%%%%%%%%%
%%%%%%%%%%%%%%%%%%%%%%%%
\section{Voting with Predictions}\label{sec:voting}
%%%%%%%%%%%%%%%%%%%%%%%%
%%%%%%%%%%%%%%%%%%%%%%%%

When only ordinal information about the preferences of the voters over the candidates is available, the best possible distortion of voting mechanisms is $\Theta(m^2)$ and the upper bound is achieved by {\sc Plurality}, which simply outputs the candidate that is ranked first by most voters \citep{caragiannis2011voting,caragiannis2017subset}. We investigate whether improvements are possible when we are given different types of predictions on top of the ordinal preferences of the voters, and show tight tradeoffs on the consistency and robustness of voting mechanisms. 

%%%%%%%%
\subsection{Predicting the Optimal Candidate}
%%%%%%%%
We start with the case where the prediction is the optimal candidate (i.e., the candidate that maximizes the social welfare), and show that such a prediction is very weak and cannot lead to any meaningful improvement in the distortion. 

\begin{theorem} \label{thm:voting:optimal-alternative-prediction}
Given the optimal candidate as prediction, there is no mechanism that simultaneously achieves consistency $o(m^2)$ and bounded robustness. 
\end{theorem}

\begin{proof}
Consider an instance with $n$ voters and $m$ candidates $\{a_1, \ldots, a_{m-1},o\}$. 
For each $\ell \in [m-1]$, there is a set $S_\ell$ of $n/(m-1)$ voters that all rank candidate $a_\ell$ first, $o$ second, and then have an arbitrary ordering of the remaining candidates. In addition, the predicted optimal candidate is $o$.

Suppose that a mechanism outputs $o$ as the winner. Clearly, such a mechanism has consistency $1$ in case $o$ is indeed the optimal candidate. However, the robustness of the mechanism is unbounded. Indeed, consider the valuation profile according to which all voters have value $1$ for their top-ranked candidates and $0$ for all others. Then, the social welfare of $o$ is $0$, while the social welfare of any other candidate is $n/(m-1)$, leading to infinite robustness. So, to achieve bounded robustness, the mechanism must choose candidate $a_\ell$ for some $\ell \in [m-1]$. Given this, consider the valuation profile according to which all the voters in set $S_\ell$ have value $1/m$ for all candidates, and all other voters have value $1/2$ for their two most-preferred candidates and $0$ for all others. Hence, the social welfare of $a_\ell$ is $\frac{n}{m(m-1)}$, whereas the social welfare of $o$ is
$$\frac{n}{m(m-1)} + \left(n-\frac{n}{m-1}\right)\cdot \frac12 \geq \frac{n}{2} \cdot \frac{m-1}{m},$$ 
thus leading to a consistency of $\Omega(m^2)$. 
\end{proof}

%%%%%%%%
\subsection{Predicting the Valuation Profile}
%%%%%%%%
Since getting the optimal candidate as a prediction does not lead to an improved distortion compared to just using the available ordinal preferences of the voters, we now switch to predictions about the valuation profile of the voters. We again show an impossibility, which however leaves room for achieving consistency $\omega(m^2)$ while at the same time not losing too much on robustness.

\begin{theorem}\label{thm:voting:full-valuation:lower}
Given the full valuation profile as prediction, there is no mechanism that simultaneously achieves consistency strictly smaller than $m-1$ and bounded robustness.
\end{theorem}

\begin{proof}
Consider an instance similar to that in the proof of \cref{thm:voting:optimal-alternative-prediction} with $n$ voters and $m$ candidates $\{a_1, \ldots, a_{m-1},o\}$. For each $\ell \in [m-1]$, there is a set $S_\ell$ of $n/(m-1)$ voters that all rank candidate $a_\ell$ first, $o$ second, and then have an arbitrary ordering of the remaining candidates. In addition, the predicted valuation profile $\hat{\bv}$ is such that all voters have value $1/2$ for each of their two most-preferred candidates, and $0$ for all others. 
    
Suppose first that $\hat{\bv}$ is accurate. Then, candidate $o$ is the optimal one with social welfare $n/2$, while the social welfare of any other candidate $a_\ell$ is $\frac{n}{2(m-1)}$. If the mechanism chooses $a_\ell$ as the winner, its consistency is at least $m-1$. Hence, it must be the case that the mechanism chooses $o$ as the winner so that it achieves consistency strictly better than $m-1$. However, $\hat{\bv}$ might be inaccurate and the true valuation profile $\bv$ might be such all voters have value $1$ for their top-ranked candidates and $0$ for all others. Hence, the social welfare of the winner $o$ is $0$, whereas the social welfare of any other candidate is $n/(m-1)$, leading to unbounded robustness. 
\end{proof}

We next present a mechanism that achieves consistency linear in the number $m$ of candidates and at the same time bounded robustness which is not that much worse than the worst-case distortion guarantee of the best ordinal mechanism (that does not use any predictions at all). For any candidate $x$, let $T_x$ be the set of voters that rank $x$ first; note that these sets only depend on the ordinal preferences of the voters. Let $\SW_1(x|\bv) = \sum_{i \in T_x} v_i(x)$ be the total value of the voters in $T_x$ for $x$; in other words, $\SW_1(x|\bv)$ is the {\em first-place social welfare} of $x$ according to the valuation profile $\bv$. In addition, let $\plu(x) = |T_x|$ be the {\em plurality score} of $x$. Our mechanism outputs the candidate with maximum predicted first-place social welfare. See Mechanism~\ref{mech:max-predicted-SW1}.

\SetCommentSty{mycommfont}
\begin{algorithm}[ht]
\SetNoFillComment
\caption{}
\label{mech:max-predicted-SW1}
\SetKwInOut{Input}{Input}
\Input{Ordinal profile $\succ$, predicted valuation profile $\hat{\bv}$}
\For{each $x\in A$}
{
    $T_x \gets $ voters that rank $x$ first\;
    $\SW_1(x|\hat{\bv}) \gets \sum_{i \in T_x} \hat{v}_i(x)$\;
}
\Return $\arg\max_{x \in A} \SW_1(x|\hat{\bv})$\; 
\end{algorithm}

Before analyzing the distortion guarantees of our mechanism, we prove a technical lemma with properties that are implied by the definition of the first-place social welfare, the definition of plurality score, and the unit-sum assumption. These will be useful in many parts of the following proofs in this section. 

\begin{lemma}\label{lem:voting:unit-sum:implications}
For any valuation profile $\bv$ and candidate $x \in A$, the following are true:
\begin{enumerate}
    \item[(i)]  $\SW(x | \bv) \leq m \cdot \max_{y \in A} \SW_1(y|\bv)$;
    \item[(ii)] $\SW_1(x|\bv) \leq \plu(x)$;
    \item[(iii)] $\frac{1}{m} \cdot \plu(x) \leq \SW_1(x|\bv) \leq \SW(x|\bv)$. 
\end{enumerate}
\end{lemma}

\begin{proof}
For (i), since $\max_{z \in A} v_i(z) = v_i(y)$ if $i \in T_y$ and $\SW_1(y | \bv) = \sum_{i \in T_y} v_i(y)$, we have
\begin{align*}
\SW(x|\bv) 
    = \sum_{i \in N} v_i(x) 
    &\leq \sum_{i \in N} \max_{z \in A} v_i(z) \\
    &= \sum_{y \in A} \sum_{i \in T_y} v_i(y)
    = \sum_{y \in A} \SW_1(y|\bv) 
    \leq m \cdot \max_{y \in A} \SW_1(y|\bv).
\end{align*}
For (ii) and (iii), by the unit-sum assumption, for any $i \in T_x$, we have that $\frac{1}{m} \leq v_i(x) \leq 1$. Hence, 
\begin{align*}
    \SW_1(x|\bv) = \sum_{i \in T_x} v_i(x) \leq \plu(x)
\end{align*}
and 
\begin{align*}
    \frac{1}{m} \cdot \plu(x) \leq \sum_{i \in T_x} v_i(x) = \SW_1(x|\bv) \leq \SW(x|\bv).
\end{align*}
This completes the proof.
\end{proof}

We are now ready to analyze the consistency and robustness of Mechanism~\ref{mech:max-predicted-SW1}.

\begin{theorem} \label{thm:voting:upper:mech:max-predicted-SW1}
Mechanism~\ref{mech:max-predicted-SW1} achieves consistency at most $m$ and robustness at most $\min\{nm,m^3\}$.
\end{theorem}

\begin{proof}
We first prove the consistency bound of $O(m)$ for when the prediction is accurate, that is, $\hat{\bv} = \bv$. 
By definition, the mechanism chooses a candidate $w$ with maximum first-place social welfare, and thus 
$\SW(w |\bv) \geq \SW_1(w |\bv) = \max_{y \in A} \SW_1(y | \bv)$. We can now bound the social welfare of the optimal candidate $o$ by applying \cref{lem:voting:unit-sum:implications}~(i) for $x=o$: 
\begin{align*}
    \SW(o|\bv) \leq m \cdot \max_{y\in A}\SW_1(y | \bv) \leq m \cdot \SW(w|\bv). 
\end{align*}

We next consider the case where the predicted valuation profile $\hat{\bv}$ may be arbitrarily inaccurate with respect to the true valuation profile $\bv$. We show two bounds on the robustness of the mechanism: $nm$ and $Om^3$. For the first bound, observe that, by the unit-sum assumption, each voter has value at most $1$ for any candidate, and thus the social welfare of an optimal candidate $o$ is $\SW(o | \bv) \leq n$. Since the mechanism maximizes the predicted first-place social welfare, it must be the case that $\plu(w) \geq 1$, and thus, by \cref{lem:voting:unit-sum:implications}~(iii), we have that $\SW(w|\bv) \geq 1/m$. Hence, the robustness is at most $nm$. 

For the second bound of $m^3$, observe that the number $n$ of agents is equal to the total plurality score of all candidates. Hence, we can bound the social welfare of an optimal candidate $o$ as follows:
\begin{align*}
    \SW(o | \bv) \leq n = \sum_{x \in A} \plu(x). 
\end{align*}
By the definition of the mechanism, the chosen candidate $w$ is such that $\SW_1(w | \hat{\bv} ) \geq \SW_1(x | \hat{\bv})$ for any candidate $x$. 
Hence, by applying \cref{lem:voting:unit-sum:implications}~(iii) for all candidates, the aforementioned property of $w$, and then \cref{lem:voting:unit-sum:implications}~(ii) for $x=w$, we obtain
\begin{align*}
    \frac{1}{m} \cdot \sum_{x \in A} \plu(x) \leq \sum_{x \in A} \SW_1(x | \hat{\bv}) \leq m \cdot \SW_1(w | \hat{\bv} ) \leq m \cdot \plu(w)
\end{align*}
This implies
\begin{align*}
    \SW(o | \bv) \leq \sum_{x \in A} \plu(x) \leq m^2 \cdot \plu(w). 
\end{align*}
The robustness bound of $m^3$ follows by the fact that $\SW(w | \bv) \geq \frac{1}{m}\cdot \plu(w)$ (due to \cref{lem:voting:unit-sum:implications}~(iii) for $x=w$), and the proof is complete.
\end{proof}

Observe that Mechanism~\ref{mech:max-predicted-SW1} achieves the best possible consistency-robustness tradeoff for instances with a small number of voters compared to candidates. In particular, if $n \leq m$, then the mechanism has consistency $O(m)$ and robustness $O(m^2)$, and this is best possible due to \cref{thm:voting:full-valuation:lower}. For instances with a small number of candidates, however, the mechanism only provides a robustness guarantee of $O(m^3)$. The main issue is that the candidate chosen by Mechanism~\ref{mech:max-predicted-SW1} might have a small plurality score, despite maximizing the predicted social welfare. A way to improve the robustness is by aiming to achieve a relatively larger plurality score, so that we achieve a better guarantee in case the prediction is inaccurate. However, by doing so, we potentially sacrifice some predicted social welfare, which can affect the consistency guarantee when the prediction is accurate. To strike a balance between the two, we define a mechanism that uses a parameter $\lambda \leq m$ and chooses the candidate with the largest plurality score among those that achieve a $\lambda$-approximation of the maximum possible predicted social welfare. 
See Mechanism~\ref{mech:parameterized-Plurality-SW} for the precise definition.
Note that by setting $\lambda = 1$, Mechanism~\ref{mech:parameterized-Plurality-SW} reduces to Mechanism~\ref{mech:max-predicted-SW1}.

\begin{algorithm}[ht]
\SetNoFillComment
\caption{}
\label{mech:parameterized-Plurality-SW}
\SetKwInOut{Input}{Input}
\Input{Ordinal profile $\succ$, predicted valuation profile $\hat{\bv}$, parameter $\lambda \in [1,m]$}
\For{each $x\in A$}
{
    $T_x \gets $ voters that rank $x$ first\;
    $\SW_1(x|\hat{\bv}) \gets \sum_{i \in T_x} \hat{v}_i(x)$\;
}
$a^* \gets \arg\max_{x \in A} \SW_1(x|\hat{\bv})$\; 
$S_\lambda \gets \bigg\{x \in A: \SW_1(x|\hat{\bv}) \geq \frac{1}{\lambda} \cdot \SW_1(a^*|\hat{\bv}) \bigg\}$\; 
\Return $w \in \arg\max_{x\in S_\lambda} \plu(x)$\;
\end{algorithm}

\begin{theorem} \label{thm:voting:upper:parameterized-Plurality-SW}
For any $\lambda \in [1,m]$, Mechanism~\ref{mech:parameterized-Plurality-SW} achieves consistency at most $\lambda m$ and robustness at most $m^3/\lambda$. 
\end{theorem}

\begin{proof}
Let $w$ be the candidate chosen by Mechanism~\ref{mech:parameterized-Plurality-SW}, and $o$ an optimal candidate. We first consider the case where $\hat{\bv} = \bv$. Since $a^*$ is the candidate with maximum first-place social welfare, by \cref{lem:voting:unit-sum:implications}~(i), we have that 
$$\SW_1(a^{*}|\bv) = \max_{y \in A} \SW_1(y|\bv) \geq \frac{1}{m} \cdot \SW(o |\bv).$$  
Using this in combination with the fact that $w \in S_{\lambda}$, we have that 
\begin{align*}
    \SW(w | \bv) \geq \SW_1(w|\bv) \geq \frac{1}{\lambda} \cdot \SW_1(a^*|\bv) \geq \frac{1}{\lambda m} \cdot \SW(o |\bv),
\end{align*}
and thus the consistency is at most $\lambda m$. 

For the robustness, since each voter has value at most $1$ for any candidate, we have
\begin{align*}
    \SW(o | \bv) \leq n = \sum_{y \in A} \plu(y) = \sum_{y \in S_\lambda} \plu(y) + \sum_{y \not\in S_\lambda} \plu(y).
\end{align*}
We now bound the plurality score of any candidate $y$ in terms of the social welfare of $w$, depending on whether $x$ belongs to $S_\lambda$, or not. 
\begin{itemize}
\item If $y \in S_\lambda$, then, since $w$ is the candidate in $S_\lambda$ with maximum plurality score, by applying \cref{lem:voting:unit-sum:implications}~(iii) for $x=w$, we have 
$$\plu(y) \leq \plu(w) \leq m \cdot \SW(w|\bv).$$ 

\item If $x \not\in S_\lambda$, then by the definition of $S_\lambda$, $\SW_1(x|\hat{\bv}) < \frac{1}{\lambda} \cdot \SW_1(a^*|\hat{\bv})$. 
By applying \cref{lem:voting:unit-sum:implications}~(iii) for $x=y$, we have that $\plu(y) \leq m\cdot \SW_1(y | \hat{\bv})$ and $\plu(y) \leq m \cdot \SW(y|\bv)$. In addition, since both $w, a^* \in S_\lambda$ and $w$ is the candidate in $S_\lambda$ with maximum plurality score, by applying \cref{lem:voting:unit-sum:implications}~(ii) for $x=a^*$, we have that $\SW_1(a^*|\hat{\bv}) \leq \plu(a^*) \leq \plu(w)$. 
Combining all of these, we get 
    \begin{align*}
        \plu(y) \leq m \cdot \SW_1(y| \hat{\bv}) < \frac{m}{\lambda} \cdot  \SW_1(a^*|\hat{\bv}) 
        \leq \frac{m}{\lambda} \cdot \plu(a^*) \leq \frac{m}{\lambda} \cdot \plu(w) \leq  \frac{m^2}{\lambda} \cdot \SW(w|\bv).
    \end{align*}
\end{itemize}
Hence, putting everything together, and using the fact that $m \leq m^2/\lambda$ (since $\lambda \leq m$), we have
\begin{align*}
    \SW(o |\bv) 
    &\leq m \cdot \sum_{y \in S_\lambda} \SW(w|\bv)  + \frac{m^2}{\lambda} \cdot \sum_{y \not\in S_\lambda}  \SW(w|\bv) \\
    &\leq \frac{m^2}{\lambda} \cdot \sum_{y \in A} \SW(w|\bv) 
    = \frac{m^3}{\lambda} \cdot \SW(w | \bv),
\end{align*}
and thus the robustness is at most $m^3/\lambda$, as claimed. 
\end{proof}

%By appropriately choosing the value of the parameter $\lambda$ used by Mechansim~\ref{mech:parameterized-Plurality-SW}, we obtain different tradeoffs between consistency and robustness. Some examples are given in the following corollary. 

%\begin{corollary} \alex{Not sure if this is useful statement; does it really add anything?}\giorgos{I agree; remove it} Mechansim~\ref{mech:parameterized-Plurality-SW} achieves 
%\begin{itemize}
%    \item consistency $O(m)$ and robustness $O(m^3)$ using $\lambda = 1$;
%    \item consistency $O(m^{1.5})$ and robustness $O(m^{2.5})$ using $\lambda = \sqrt{m}$;
%    \item consistency $O(m^2)$ and robustness $O(m^2)$ using $\lambda = m$.
%\end{itemize}
%\end{corollary}

Next, we show that the consistency-robustness tradeoff achieved by Mechanism~\ref{mech:parameterized-Plurality-SW} is asymptotically best possible over all mechanisms that take as input a prediction of the valuation profile. We remark that our mechanism actually does not require a prediction of the whole valuation profile, but only the values of the voters for the candidates they rank first.

\begin{theorem} \label{thm:voting:full-valuation:lower:tradeoff}
For any $\lambda \in [1,m]$, any mechanism with consistency $o(\lambda m)$ has robustness $\Omega(m^3/\lambda)$.
\end{theorem}

\begin{proof}
Fix any $\lambda \in [1,m]$ and consider any mechanism with consistency $o(\lambda m)$. 
We consider an instance with an odd number $m = 1+2k \geq 5$ of candidates named
$\{a, b_1, \ldots, b_{k}, c_1, \ldots, c_{k}\}$. 
The ordinal preferences of the voters are as follows:
\begin{itemize}
    \item For each $\ell \in [k]$, there is a set $B_\ell$ of $\lambda m$ voters that rank $b_\ell$ first, $a$ second, then the candidates in $\{c_1, \ldots, c_{k}\}$ arbitrarily, and then the candidates in $\{b_1, \ldots, b_{k}\}\setminus\{b_\ell\}$ arbitrarily.
    
    \item For each $\ell \in [k]$, there is a set $C_\ell$ of $m^2$ voters that rank $c_\ell$ first, $a$ second, then the candidates in $\{b_1, \ldots, b_{k}\}$ arbitrarily, and then the candidates in $\{c_1, \ldots, c_{k}\}\setminus\{c_\ell\}$ arbitrarily.
\end{itemize}
Since $|\bigcup_{\ell \in [k]}\{B_\ell\}| = k \lambda m$ and $|\bigcup_{\ell \in [k]}\{C_\ell\}| = k m^2$, our instance consists of $n = km (\lambda + m)$ voters.
The predicted valuation profile $\hat{\bv}$ is as follows:
\begin{itemize}
    \item For each $\ell \in [k]$, the voters in $B_\ell$ have value $1/2$ for $b_\ell$, $1/2$ for $a$, and $0$ for the remaining candidates;
    \item For each $\ell \in [k]$, the voters in $C_\ell$ have value $\frac{1}{k+2}$ for each candidate in $\{c_\ell, a, b_1, \ldots, b_{k}\}$, and $0$ for the remaining candidates. 
\end{itemize}
Since $\frac{1}{2k} \leq \frac{1}{k+2} \leq \frac{1}{k}$ and $k = (m-1)/2 \geq m/4$, we have
\begin{align*}
    \SW(a | \hat{\bv}) &\geq k \lambda m \cdot \frac12 + k m^2\cdot \frac{1}{2k} \geq \frac18 \lambda m^2\\
    \SW(b_\ell | \hat{\bv}) &\leq \lambda m \cdot \frac12 + k m^2\cdot \frac{1}{k} \leq 2m^2, \forall \ell \in [k]\\ 
    \SW(c_\ell | \hat{\bv}) &\leq m^2\cdot \frac{1}{k} \leq 4m, \forall \ell \in [k]
\end{align*}
Clearly, if the mechanism chooses any candidate $c_\ell$, then its consistency would be at least $\lambda m /16$. So, to achieve consistency $o(\lambda m)$, it has to be the case that the mechanism chooses either the optimal predicted candidate $a$, or a candidate $b_\ell$ for some $\ell \in [k]$. We consider each case separately and show a lower bound on the robustness of the mechanism.

\medskip
\noindent 
{\bf Case 1: The mechanism chooses $a$.}
The true valuation profile $\bv$ might be such that all voters have value $1$ for the candidates they rank first and $0$ for the remaining ones. Hence, since $\plu(a) = 0$, $\SW(a | \bv) = 0$, while any other candidate $x \neq a$ has $\SW(x,|\bv) > 0$, leading to unbounded robustness. 

\medskip
\noindent 
{\bf Case 2: The mechanism chooses $b_\ell$ for some $\ell \in [k]$.}
The true valuation profile $\bv$ might be as follows: 
\begin{itemize}
    \item All voters in $B_\ell$ have value $1/m$ for all candidates;
    \item All other voters have value $1/2$ for their two most-preferred candidates. 
\end{itemize}
Under this profile, $b_\ell$ only gets value $1/m$ from each of the $\lambda m$ candidates in $B_\ell$, and thus $\SW(b_\ell | \bv) = \lambda$. One the other hand, since $a$ is ranked second by all voters in $\bigcup_{\ell \in [k]}\{C_\ell\}$ and $k \geq m/4$, its social welfare is $\SW(a|\bv) \geq k m^2 \cdot \frac12 \geq \frac{m^3}{8}$. Hence, the robustness is at least $m^3/(8\lambda)$.
\end{proof}

\subsubsection*{Prediction Error}
We conclude the section with an analysis of the distortion of Mechanism~\ref{mech:parameterized-Plurality-SW} as a function of a prediction error, which shows how the distortion changes from consistency (when there is no prediction error) to robustness (when the error is as large as possible). We first define what we precisely mean by prediction error for our problem. For any candidate $x \in A$ and valuation profile $\bv$, let 
\begin{align*}
    \rho(x|\bv) = \frac{\max_{y \in A} \SW_1(y|\bv)}{\SW_1(x|\bv)}
\end{align*}
be the approximation of the maximum first-place social welfare achieved by $x$ with respect to $\bv$. When $\bv$ is the true (unknown) valuation profile and $\hat{\bv}$ is the given predicted valuation profile, we define the prediction error $\eta(x)$ of a candidate $x$ as
\begin{align*}
    \eta(x) = \max \left\{ \frac{\rho(x|\bv)}{\rho(x|\hat{\bv})}, 1 \right\}.
\end{align*}
Essentially, the prediction error of $x$ quantifies (in a multiplicative way) how far away the true approximation of the maximum first-place social welfare achieved by $x$ is from the predicted approximation that $x$ achieves. Since our mechanism chooses a candidate from the set $S_\lambda = \{x \in A: \rho(x|\hat{\bv}) \leq \lambda\}$ of candidate that achieve a $\lambda$-approximation of the predicted maximum first-place social welfare, we are only in the worst prediction error of those candidates. Hence, we define $\eta := \max_{x \in S_\lambda} \eta(x)$. 

\begin{theorem} \label{thm:voting:error-analysis}
For any $\lambda \in [1,m]$, $\rho \in [1,\lambda]$ and $\eta \in \left[1, \frac{m^2}{\lambda \cdot \rho}\right]$, the distortion of Mechanism~\ref{mech:parameterized-Plurality-SW} is $m \cdot \eta \cdot \rho$.
\end{theorem}

\begin{proof}
By \cref{lem:voting:unit-sum:implications}~(i), the definition of $\rho(w|\bv)$, and since $\SW_1(w|\bv) \leq \SW(w|\bv)$, we have
\begin{align*}
    \SW(o | \bv) 
    &\leq m \cdot \max_{x \in A} \SW_1(x | \bv) \\
    &= m \cdot \rho(w | \bv) \cdot \SW_1(w | \bv) \\
    &= m \cdot \frac{\rho(w | \bv) }{\rho(w | \hat{\bv}) } \cdot \rho(w | \hat{\bv}) \cdot \SW_1(w | \bv) \\
    &\leq m \cdot \eta(w) \cdot \rho(w | \hat{\bv}) \cdot \SW(w | \bv).
\end{align*}
To bound the possible values of $\eta(w)$, let $a \in \arg\max_{x \in A} \SW_1(x|\bv)$ be a candidate that achieves the maximum true first-place social welfare. Since $\SW_1(a|\bv) \leq \plu(a)$ (by \cref{lem:voting:unit-sum:implications}~(ii)) and $\SW_1(w |\bv) \geq \frac{1}{m} \plu(w)$ (By \cref{lem:voting:unit-sum:implications}~(iii)), we have that $\rho(w|\bv) \leq m \cdot \frac{\plu(a)}{\plu(w)}$. Hence, 
\begin{align*}
\eta(w) = \frac{\rho(w|\bv)}{\rho(w|\hat{\bv})} \leq \frac{m \cdot \plu(a)}{\plu(w) \cdot \rho(w|\hat{\bv})}.
\end{align*}
We now switch between the following cases:
\begin{itemize}
    \item $\plu(a) \leq \plu(w)$. Since $\rho(w|\hat{\bv}) \geq 1$, we have that $\eta(w) \leq m$. 
    
    \item $\plu(a) > \plu(w)$. It must be the case that $a \not\in S_\lambda$, and thus 
    $\SW_1(a | \hat{\bv}) < \frac{1}{\lambda} \SW_1(a^* | \hat{\bv})$. Since $w$ is the candidate with maximum plurality score in $S_\lambda$ and $a^* \in S_\lambda$, by \cref{lem:voting:unit-sum:implications}~(ii) and~(iii), we have
    \begin{align*}
        \frac{1}{m} \plu(a) < \SW_1(a | \hat{\bv}) < \frac{1}{\lambda} \SW_1(a^* | \hat{\bv}) \leq \frac{1}{\lambda} \plu(a^*) \leq \frac{1}{\lambda} \plu(w),
    \end{align*}
    or, equivalently, $\plu(a) \leq \frac{m}{\lambda} \cdot \plu(w)$, which further implies that $\eta(w) \leq \frac{m^2}{\lambda \cdot \rho(w|\hat{\bv})}$. 
\end{itemize}
Since clearly $\rho(w|\hat{\bv}) \in [1,\lambda]$ by the definition of $w$, the distortion of the mechanism is $O(m \cdot \eta \cdot \rho)$ for any $\rho \in [1,\lambda]$ and $\eta \in \left[1, \frac{m^2}{\lambda \cdot \rho}\right]$. 
\end{proof}

Observe that \cref{thm:voting:error-analysis} recovers the consistency and the robustness bound of Mechanism~\ref{mech:parameterized-Plurality-SW}. In particular, when the prediction is accurate, then $\eta = 1$, and since $\rho \leq \lambda$, we get the consistency bound $O(\lambda m)$. When the prediction is inaccurate, then $\eta$ takes its maximum value of $\frac{m^2}{\lambda \cdot \rho}$, and we get the robustness bound of $O(m^3/\lambda)$. 

%%%%%%%%%%%%%%%%%%%%%%%%%%%%%%%%%%%%%%%%%%%%%%%%%%%%%%%%%%%%%
%%%%%%%%%%%%%%%%%%%%%%%%%%%%%%%%%%%%%%%%%%%%%%%%%%%%%%%%%%%%%
\section{One-sided Matching with Predictions}\label{sec:matching}
%%%%%%%%%%%%%%%%%%%%%%%%%%%%%%%%%%%%%%%%%%%%%%%%%%%%%%%%%%%%%
%%%%%%%%%%%%%%%%%%%%%%%%%%%%%%%%%%%%%%%%%%%%%%%%%%%%%%%%%%%%%
We now switch to the one-sided matching setting. Recall that here the goal is to compute a matching between a set of agents and a set of items. In the single-winner voting setting, we established that, even when given a prediction of the whole valuation profile, it impossible to achieve sublinear consistency while maintaining bounded robustness; see \cref{thm:voting:full-valuation:lower}. Contrary to that, in the one-sided matching setting, it turns out that given a prediction about the whole valuation profile is extremely powerful. In particular, it allows us to achieve a consistency of $1$ (i.e., we compute an optimal matching when the prediction is accurate), while still guaranteeing a robustness of $n^2$, which asymptotically matches the best possible distortion of ordinal mechanisms \citep{amanatidis2022few}.

The mechanism that achieves this optimal tradeoff simply chooses a matching that maximizes the predicted social welfare such that some item is matched to an agent that ranks it first. The following lemma establishes that such an optimal matching always exists, and it can be computed in polynomial time. 

\begin{lemma}\label{lem:opt-assigns-top-item}
Let $\mu^{*} \in \arg\max_{\mu} \SW(\mu | \bv)$ be any optimal matching according to some valuation profile $\bv$. Then, in polynomial time, $\mu^{*}$ can be transformed into another optimal matching $\hat{\mu}^{*} \in \arg\max_{\mu} \SW(\mu | \bv)$ such that for some agent $i \in N$, $\hat{\mu}^{*}(i)$ is $i$'s most-preferred item. 
\end{lemma}

\begin{proof}
If there exists some agent $i \in N$ such that $\mu^*(i)$ is her most-preferred item, then $\hat{\mu}^{*}=\mu^{*}$ and we are done. Therefore, suppose that, in $\mu^{*}$, every agent is matched to an item that she ranks as second or worse. Consider a graph $G$ in which there is a node $i$ for every agent $i \in N$, and an edge $(i,j)$ indicates that agent $\mu^{*}(j)$ is agent $i$'s most-preferred item. Since for every agent $i$ there is another agent $j$ that is matched to $i$'s most-preferred item, $G$ contains at least one cycle. 
We can now construct a new matching $\hat{\mu}^{*}$ by assigning to the agents that are involved in a cycle their optimal items. That is, take any cycle in $G$ and if $(i,j)$ is an edge of that cycle, then set $\hat{\mu}^{*}(i) = \mu^{*}(j)$.
Since $\mu^*$ and $\hat{\mu}^{*}$ differ only on the items that are matched to the agents of the considered cycle, and these agents are matched to their most-preferred items in $\hat{\mu}^{*}$, we have that $\SW(\hat{\mu}^{*}|\bv) \geq \SW(\mu^{*}|\bv)$, implying that $\hat{\mu}^*$ is optimal as well.   
\end{proof}

\begin{remark}\label{rem:optimal-matching}
Given \cref{lem:opt-assigns-top-item}, in the rest of this section, when we refer to an optimal matching (or a partial optimal matching), we will assume that at least one agent is matched to her most-preferred item. Furthermore, if this matching is optimal with respect to a predicted valuation profile $\hat{\bv}$, it will maintain this property with respect to the true valuation profile $\bv$, since the ordinal preferences of the agents are induced by both $\bv$ and $\hat{\bv}$ by definition. 
\end{remark}

Using \cref{lem:opt-assigns-top-item}, the following theorem is almost immediate.

\begin{theorem}\label{thm:matching-full-valuation-profile-as-prediction}
Given as prediction a full valuation profile $\hat{\bv}$, the mechanism that outputs an optimal matching with respect to $\hat{\bv}$ achieves consistency $1$ and robustness at most $n^2$.  
\end{theorem}

\begin{proof}
Clearly, if the predicted valuation profile $\hat{\bv}$ is accurate, then the mechanism outputs an optimal matching according to the true valuation profile $\bv$, and thus the consistency is $1$. For the robustness, when $\hat{\bv} \neq \bv$, by \cref{lem:opt-assigns-top-item}, the matching $\hat{\mu}^*$ that the mechanism outputs is such that at least one agent is matched to her most-preferred item. By the unit-sum normalization, this implies that $\SW(\hat{\mu}^*|\bv) \geq 1/n$, and since $\SW(\mu|\bv) \leq n$ for any matching $\mu$, the robustness bound follows. 
\end{proof}

Of course, having access to the whole valuation profile as a prediction may be too generous. Given this, we next consider the case where only part of the valuation profile can be predicted, namely the values of the agents for the items they rank in the first $k$ positions, for some $k \in [n]$. To simplify our notation, we denote by $\top_\ell(i)$ the item that agent $i$ ranks at position $\ell \in [m]$, and by $\bv_k$ the {\em $k$-truncated valuation profile} that consists of the value $v_i(x)$ for every agent $i$ for any $x \in \bigcup_{\ell=1}^x \top_\ell(i)$. 

We first show a positive result by appropriately modifying the aforementioned mechanism. In particular, given an ordinal profile $\succ$ and as prediction a $k$-truncated valuation profile $\hat{\bv}_k$, our mechanism first completes $\hat{\bv}_k$ by setting the undefined values to $0$, and then outputs an optimal matching according to that valuation profile. See Mechanism~\ref{mech:k-welfare-max}. The next theorem shows that the consistency improves linearly in $k$, without compromising any robustness.

\SetCommentSty{mycommfont}
\begin{algorithm}[ht]
\SetNoFillComment
\caption{}
\label{mech:k-welfare-max}
\SetKwInOut{Input}{Input}
\Input{Ordinal profile $\succ$, predicted $k$-truncated valuation profile $\hat{\bv}_k$}
\For{each agent $i \in N$ and $\ell \in \{k+1,\ldots,n\}$}
{
$\hat{v}_{i, \top_\ell(i)} \gets 0$\;
}
\Return $\arg\max_{\mu} \SW(\mu|\hat{\bv})$\; 
\end{algorithm}

\begin{theorem}\label{thm:matching:upper}
Mechanism~\ref{mech:k-welfare-max} achieves consistency at most $n/k+2$ and robustness at most $n^2$.
\end{theorem}

\begin{proof}
Let $\mu^{*}$ be the optimal matching according to the true valuation profile $\bv$, 
and let $\mu$ be the matching computed by the mechanism when given the $k$-truncated valuation profile $\hat{\bv}_k$ as prediction. First consider the case where the prediction is accurate, that is, when $\hat{\bv}_k = \bv_k$. 
Let $S = \{i \in N: \mu^*(i) \not\in \bigcup_{\ell=1}^k \top_\ell(i) \}$ be the set of agents whose optimal items are not among those they rank in the first $k$ positions. Clearly, if $S=\varnothing$, then $\mu$ is by construction a welfare-maximizing matching and the consistency is $1$.

We now partition the set $S$ into $\lceil |S|/k \rceil$ sets $G_1, \ldots, G_{\lceil |S|/k \rceil}$ of agents of size $k$ each, with the exception of the last group which may have size smaller than $k$. For each set $G \in \{G_1, \ldots, G_{\lceil |S|/k \rceil} \}$, there exists a partial matching $\mu_G$ that involves only the agents in $G$ and assigns to them (at most) $k$ items, such that $v_i(\mu_G(i)) \geq v_i(\mu^{*}(i))$ for every agent $i \in G$.
The existence of this matching is guaranteed by the fact that there are at most $k$ agents in $G$, and none of them is matched to any of the items she ranks in the first $k$ positions in $\mu^{*}$. 
By construction, $\mu$ achieves at least as much social welfare as any other matching when restricted to the first $k$ positions, and hence we have that 
\begin{align*}
\SW(\mu|\bv) \geq \sum_{i \in G} v_i(\mu_G(i)) \geq \sum_{i \in G} v_i(\mu^{*}(i)),
\end{align*}
By Summing over all these inequalities (for all groups) we obtain that
\begin{align*}
\left(\frac{|S|}{k}+1\right)\SW(\mu|\bv) \geq \left\lceil\frac{|S|}{k}\right\rceil\SW(\mu|\bv) \geq \sum_{i \in S} v_i(\mu^{*}(i)).
\end{align*}
Since $\mu$ is the optimal matching according to $\bv_k$, we also have that
\begin{align*}
\SW(\mu|\bv) \geq \sum_{i \in N\setminus S} v_i(\mu^{*}(i)). 
\end{align*}
Putting everything together, we obtain
\begin{align*}
    \SW(\mu^*|\bv) \leq \left( \frac{|S|}{k}+2 \right)\cdot \SW(\mu|\bv),
\end{align*}
and the consistency bound follows since $|S| \leq n$.

For the robustness, the argument is similar to that in the proof of \cref{thm:matching-full-valuation-profile-as-prediction}. By the definition of the mechanism, the computed matching $\mu$ is such that some agent is matched to her most-preferred item. Hence, the achieved social welfare is at least $1/n$. Since the maximum possible social welfare is at most $n$, the robustness bound follows. 
\end{proof}

Before we proceed, we parameterize the distortion guarantee of Mechanism~\ref{mech:k-welfare-max} as a function of an appropriately defined prediction error, as we did in Section~\ref{sec:voting} for Mechanism~\ref{mech:parameterized-Plurality-SW} (\cref{thm:voting:error-analysis}).
Let $\mu^*_k \in \max_\mu \SW(\mu|\bv_k)$ and $\hat{\mu}^*_k \in \max_\mu \SW(\mu|\hat{\bv}_k)$ be the (partial) matchings that maximize the social welfare according to the true and the predicted $k$-truncated valuation profiles, assuming that the undefined are $0$, respectively. We define the prediction error as the ratio of the true social welfare achieved by these two matchings, that is, 
\begin{align*}
    \eta = \frac{\SW(\mu_k^*|\bv)}{\SW(\hat{\mu}_k^*|\bv)}.
\end{align*}
Intuitively, the prediction error measures how far the predicted optimal matching is from the actual optimal matching, when restricted to the $k$ first positions of the agents rankings. 

\begin{theorem}\label{thm:matching:upper:error-analysis}
For any $\eta \in [1,n^2]$, the distortion of Mechanism~\ref{mech:k-welfare-max} is $O(\min\{\frac{n}{k}\cdot\eta, n^2 \})$.
\end{theorem}

\begin{proof}
Let $\mu^*$ be the optimal matching according to the true valuation profile $\bv$. 
We partition the agents into $\lceil n/k \rceil$ groups $G_1, \ldots, G_{n/k}$ of size $k$ each, with the only exception of the last group whose size might be smaller than $k$. Clearly, for every group $G \in \{G_1, \ldots, G_{n/k}\}$, the optimal matching $\mu_G$ of the agents in $G$ to the items that they rank in the first $k$ positions is such that 
\begin{align*}
    \sum_{i \in G} v_i(\mu_G(i)) \geq \sum_{i \in G} v_i(\mu^*(i)).
\end{align*}
Due to the definition of $\mu_k^*$ and by summing over all groups, we have that 
\begin{align*}
\left(\frac{n}{k}+1\right)\SW(\mu^*_k|\bv) \geq \left\lceil\frac{n}{k}\right\rceil\SW(\mu^*_k|\bv) \geq 
\sum_{G} \sum_{i \in G} v_i(\mu_G(i)) \geq \sum_{i \in N} v_i(\mu^*(i))
\end{align*}
Consequently, by the definition of $\eta$ and the definition of the mechanism, which outputs the matching $\mu = \hat{\mu}^*_k$, we obtain
\begin{align*}
\SW(\mu^*|\bv) \leq \left(\frac{n}{k}+1\right)\SW(\mu^*_k|\bv) = \left(\frac{n}{k}+1\right) \cdot \eta \cdot \SW(\mu|\bv).
\end{align*}
By the unit-sum normalization, we also clearly have that $\SW(\mu|\bv) \geq 1/n$ since at least one agent is matched to her most-preferred item (see Remark~\ref{rem:optimal-matching}), and $\SW(\mu^*|\bv), \SW(\mu^*_k|\bv) \leq n$. Therefore, the distortion is $O(\min\{\frac{n}{k}\cdot\eta, n^2\})$ for any $\eta \in [1,n^2]$, as claimed. 
\end{proof}

Observe that \cref{thm:matching:upper:error-analysis} recovers the consistency and robustness guarantees of Mechanism~\ref{mech:k-welfare-max}. The robustness is clearly at most $n^2$. When the precidion is accurate, then $\eta=1$, and the consistency is $n/k$.

Next, we show that the consistency-robustness tradeoff achieved by Mechanism~\ref{mech:k-welfare-max} is asymptotically best possible for any mechanism that is given a $k$-truncated valuation profile as prediction. The lower bound of $\Omega(n^2)$ on the consistency of such mechanisms follows directly from a corresponding lower bound on the distortion of purely ordinal mechanisms (which ignore the prediction altogether) shown by \citet{amanatidis2022few}.  

\begin{theorem}[\citep{amanatidis2022few}(Theorem 1)]\label{thm:matching:lower:robustness}
For any $k \in [n]$, the robustness of any matching mechanism that is given a $k$-truncated valuation profile as prediction is $\Omega(n^2)$. 
\end{theorem}

For the consistency, we have the following statement. 

\begin{theorem}\label{thm:matching:lower:consistency}
For any $k \in [n]$, the consistency of any matching mechanism that is given a $k$-truncated valuation profile as prediction is $\Omega(n/k)$.
\end{theorem}

\begin{proof}
The lower bound is straightforward for any $k > n/2$, so consider an instance with an even number $n \geq 2k$ of agents and $n$ items called 
$\{a_1, \ldots, a_k, b_1, \ldots, b_{n/2}, c_1, \ldots, c_{n/2-k}\}.$ 
The ordinal profile is such that, for any odd $i \in [n-1]$, agents $i$ and $i+1$ have the same ranking defined as:
$$a_1 \succ_i \ldots \succ_i a_k \succ_i b_{\ceil{i/2}} \succ_i [\{b_1, \ldots, b_{n/2}, c_1, \ldots, c_{n/2-k}\}\setminus \{b_{\ceil{i/2}}\}],$$
where $[X]$ denotes any arbitrary ranking of the elements of set $X$. The $k$-truncated valuation profile $\hat{\bv}_k$ which is given as prediction is such that each agent has value $1/(k+1)$ for each of the $k$ items she ranks first. To show a lower bound on the consistency, we assume that $\hat{\bv_k}$ is the true $k$-truncated valuation profile $\bv_k$.

Clearly, since all agents rank the items $a_1, \ldots, a_k$ in the same order in the first $k$ positions, we can assume without loss of generality, that these items are assigned to the last $k$ agents, that is, each agent $i \in \{n-k+1, \ldots, n\}$ is matched to one of these $k$ items. Observe now that, for any odd $i \in [n-k]$, agents $i$ and $i+1$ rank the same item $b_{\ceil{i/2}}$ at position $k+1$, but only (at most) one of them can be matched to it in any matching. To complete the valuation profile $\bv$, we set the values of agents $i$ and $i+1$ depending on whether they are assigned to $b_{\ceil{i/2}}$ in the matching computed by the mechanism:
\begin{itemize}
    \item Agent $i$ is matched to $b_{\ceil{i/2}}$: Agent $i$ has value $\frac{1}{(n-k)(k+1)}$ for each item she ranks at positions $\ell \in \{k+1, \ldots, n\}$, whereas agent $i+1$ has value $1/(k+1)$ for $b_{\ceil{i/2}}$ and $0$ for the remaining items. 
    
    \item Agent $i+1$ is matched to $b_{\ceil{i/2}}$: Agent $i$ has value $1/(k+1)$ for $b_{\ceil{i/2}}$ and $0$ for the remaining items, whereas agent $i+1$ has value $\frac{1}{(n-k)(k+1)}$ for each item she ranks at positions $\ell \in \{k+1, \ldots, n\}$.
    
    \item Neither agent $i$ nor agent $i+1$ is matched to $b_{\ceil{i/2}}$: Both agent $i$ and agent $i+1$ have value $1/(k+1)$ for $b_{\ceil{i/2}}$ and $0$ for the remaining items.
\end{itemize}
Given this valuation profile $\bv$, any possible matching $\mu$ that is computed by the mechanism has a social welfare of
\begin{align*}
    \SW(\mu|\bv) \leq \frac{k}{k+1} + \frac{n-k}{(n-k)(k+1)} = 1.
\end{align*}
On the other hand, there is another matching $\mu^*$ that, for each odd $i \in [n-k]$, assigns item $b_{\ceil{i/2}}$ to an agent among $i$ and $i+1$ that has value $1/(k+1)$ for it, while the other is assigned to an arbitrary item from the set 
$\{c_1, \ldots, c_{n/2-k}\}$. The social welfare of this matching is
\begin{align*}
    \SW(\mu^*|\bv) = \frac{k}{k+1} + \frac{n-k}{2}\cdot \frac{1}{k+1} + \frac{n-k}{2} \cdot \frac{1}{(n-k)(k+1)} \geq \frac{n+1}{2(k+1)},
\end{align*}
thus establishing the lower bound of $\Omega(n/k)$ on the consistency. 
\end{proof}

%%%%%%%%%%%%%%%%%%%%%%%%
%%%%%%%%%%%%%%%%%%%%%%%% 
\section{Conclusion and Open Problems}
%%%%%%%%%%%%%%%%%%%%%%%%
%%%%%%%%%%%%%%%%%%%%%%%%
In this work, we studied the distortion of deterministic mechanisms for two fundamental social choice problems (that of single-winner voting and one-sided matching) under the recently introduced learning-augmented framework, and showed tight consistency-robustness tradeoffs for both of these settings. In particular, our results paint a complete picture about the distortion guarantees of mechanisms that are given as input the ordinal preferences of the agents over the alternative outcomes and a prediction about their underlying values. 

As a future direction, it would be interesting to consider \emph{hybrid} settings where, besides the ordinal profile, a part of the cardinal profile is provided as accurate information (or can be elicited similarly to models in prior works, e.g., \citep{ABFV21,amanatidis2022few,ABFV24,ES25}) and another part is given as unreliable predicted information. To this end, we present here an impossibility for the voting setting. The following theorem (asymptotically) strengthens \cref{thm:voting:full-valuation:lower}, as it establishes that, for a constant $k$, even if a mechanism has access to the top-$k$ values of all agents (and is given the rest of the valuation profile as a prediction), it cannot achieve sublinear consistency while maintaining bounded robustness. It further implies that, to achieve sublinear consistency $O(m^{1-\varepsilon})$ for some $\varepsilon > 0$ and bounded robustness, the mechanism must have access to $\Omega(\varepsilon \log{m})$-top values per agent. 

\begin{theorem} \label{thm:voting:queries:lower}
For any $k \in [\log{m}]$, let $\mech$ be some voting mechanism that is given as input the ordinal preferences of the agents, the true values of the agents for the candidates they rank in the first $k$ positions, and a prediction of the values of the agents for the remaining candidates. Then, $\mech $ cannot simultaneously achieve consistency $o(m/2^k)$ and bounded robustness. 
\end{theorem}

\begin{proof}
Consider any mechanism in the class defined in the statement and has consistency $o(m/2^k)$. We define the following instance consisting of $n$ voters and $m$ candidates that we call $\{a_1, \ldots, a_{m-1},o\}$. The voters are partitioned into $m-1$ sets $S_1,\ldots,S_{m-1}$ of size $n/(m-1)$ each. The ordinal profile is such that, for any $i \in [m-1]$ and $\ell \in [k+1]$, candidate $a_i$ is ranked at position $\ell$ by all voters in set $S_{i+\ell-1 \Mod{m-1}}$. Furthermore, all voters rank candidate $o$ at position $k+2$, and the remaining candidates (that do not appear in the first $k+1$ positions of their preferences) arbitrarily. See \cref{tab:voting:queries:lower} for an example with $n=8$, $m=5$, and $k=2$. 

\begin{table}[t]
    \centering
    \begin{tabular}[h!]{cc} \hline
    Set & Ranking \\\hline
    $S_1$ & $a_1 \succ a_4  \succ a_3 \succ o \succ a_2$ \\
    $S_2$ & $a_2 \succ a_1 \succ a_4 \succ o \succ a_3$ \\
    $S_3$ & $a_3 \succ a_2 \succ a_1 \succ o \succ a_4$ \\
    $S_4$ & $a_4 \succ a_3 \succ a_2 \succ o \succ a_1$ \\\hline
    \end{tabular}
    \caption{An example of the ordinal preferences in the instance considered in the proof of \cref{thm:voting:queries:lower} with $n=8$, $m=5$ and $k=2$. Each of the sets $S_1$, $S_2$, $S_3$ and $S_4$ include two voters that have the same ranking over the candidates.}
    \label{tab:voting:queries:lower}
\end{table}

The $k$ \reliable queries reveal that every voter $i$ has value $2^{-\ell}$ for the candidate that $i$ ranks at position $\ell \in [k]$. The remaining $m-k$ \unreliable queries reveal that every voter $i$ has value $2^{-(k+1)}$ for the candidates that $i$ ranks at position $k+1$ and $k+2$, and value $0$ for the remaining candidates. Let $\hat{\bv}$ be the valuation profile that consists of the aforementioned values which are revealed by the \reliable and the \unreliable queries. Observe that the values of each agent according to $\hat{\bv}$ satisfy the unit-sum assumption: The first $k+1$ values form a geometric progression that starts from $1/2$ and has a decay rate of $1/2$; hence, the sum of the first $k+1$ values is $\sum_{\ell=1}^{k+1} 2^{-\ell} = \frac{\frac12 \left(1 - 2^{-(k+1)} \right)}{1-\frac12} = 1-2^{-(k+1)}$. 
We have that 
\begin{align*}
    \SW(o | \hat{\bv}) &= n \cdot 2^{-(k+1)} \\
    \SW(a_j | \hat{\bv}) &=  \frac{n}{m-1} \cdot \sum_{\ell=1}^{k+1} 2^{-\ell} = \frac{n}{m-1}\cdot(1-2^{-(k+1)}), \forall j \in [m-1].
\end{align*}
Since, for any $j \in [m-1]$,
\begin{align*}
    \frac{\SW(o | \hat{\bv})}{\SW(a_j | \hat{\bv})} = \frac{n \cdot 2^{-(k+1)}}{\frac{n}{m-1}\cdot(1-2^{-(k+1)})} = \frac{m-1}{2^{k+1}-1},
\end{align*} 
to achieve consistency $o(m/2^k)$, the mechanism must choose candidate $o$.

Now consider a different valuation profile $\bv$ that agrees with the answers of the $k$ \reliable queries, but is such that every voter $i$ has value $2^{-k}$ for the candidate that she ranks at position $k$, and $0$ for the remaining candidates. Observe that $\bv$ satisfies the unit-sum assumption for each voter, since now the first $k$ values form a geometric progression. Since $\SW(o|\bv) = 0$ and $\SW(a_j | \bv) > 0$ for any $j \in [m-1]$, the robustness is unbounded. 
\end{proof}

%%%%%%%%%%
%%%%%%%%%%
\subsection*{Acknowledgments}
%%%%%%%%%%
%%%%%%%%%%
Aris Filos-Ratsikas was supported by the UK Engineering and Physical Sciences Research Council (EPSRC) grant EP/Y003624/1. 

\bibliographystyle{plainnat}
\bibliography{refs}

\begin{thebibliography}{52}
\providecommand{\natexlab}[1]{#1}
\providecommand{\url}[1]{\texttt{#1}}
\expandafter\ifx\csname urlstyle\endcsname\relax
  \providecommand{\doi}[1]{doi: #1}\else
  \providecommand{\doi}{doi: \begingroup \urlstyle{rm}\Url}\fi

\bibitem[Abdulkadiro{\u{g}}lu and S{\"o}nmez(1998)]{abdulkadirouglu1998random}
Atila Abdulkadiro{\u{g}}lu and Tayfun S{\"o}nmez.
\newblock Random serial dictatorship and the core from random endowments in house allocation problems.
\newblock \emph{Econometrica}, 66\penalty0 (3):\penalty0 689--701, 1998.

\bibitem[Abdulkadiro{\u{g}}lu et~al.(2005)Abdulkadiro{\u{g}}lu, Pathak, Roth, and S{\"o}nmez]{abdulkadirouglu2005boston}
Atila Abdulkadiro{\u{g}}lu, Parag~A Pathak, Alvin~E Roth, and Tayfun S{\"o}nmez.
\newblock The boston public school match.
\newblock \emph{American Economic Review}, 95\penalty0 (2):\penalty0 368--371, 2005.

\bibitem[Abramowitz et~al.(2019)Abramowitz, Anshelevich, and Zhu]{abramowitz2019awareness}
Ben Abramowitz, Elliot Anshelevich, and Wennan Zhu.
\newblock Awareness of voter passion greatly improves the distortion of metric social choice.
\newblock In \emph{Proceedings of the 15th International Conference on Web and Internet Economics ({WINE})}, pages 3--16, 2019.

\bibitem[Agrawal et~al.(2024)Agrawal, Balkanski, Gkatzelis, Ou, and Tan]{agrawal2022learning}
Priyank Agrawal, Eric Balkanski, Vasilis Gkatzelis, Tingting Ou, and Xizhi Tan.
\newblock Learning-augmented mechanism design: Leveraging predictions for facility location.
\newblock \emph{Mathematics of Operations Research}, 49\penalty0 (4):\penalty0 2626--2651, 2024.

\bibitem[ALPS()]{ALPS}
ALPS.
\newblock {Algorithms with Predictions}.
\newblock {\footnotesize\url{https://algorithms-with-predictions.github.io/}}.

\bibitem[Amanatidis et~al.(2021)Amanatidis, Birmpas, Filos{-}Ratsikas, and Voudouris]{ABFV21}
Georgios Amanatidis, Georgios Birmpas, Aris Filos{-}Ratsikas, and Alexandros~A. Voudouris.
\newblock Peeking behind the ordinal curtain: Improving distortion via cardinal queries.
\newblock \emph{Artificial Intelligence}, 296:\penalty0 103488, 2021.

\bibitem[Amanatidis et~al.(2022)Amanatidis, Birmpas, Filos-Ratsikas, and Voudouris]{amanatidis2022few}
Georgios Amanatidis, Georgios Birmpas, Aris Filos-Ratsikas, and Alexandros~A Voudouris.
\newblock A few queries go a long way: Information-distortion tradeoffs in matching.
\newblock \emph{Journal of Artificial Intelligence Research}, 74:\penalty0 227--261, 2022.

\bibitem[Amanatidis et~al.(2024)Amanatidis, Birmpas, Filos{-}Ratsikas, and Voudouris]{ABFV24}
Georgios Amanatidis, Georgios Birmpas, Aris Filos{-}Ratsikas, and Alexandros~A. Voudouris.
\newblock Don't roll the dice, ask twice: The two-query distortion of matching problems and beyond.
\newblock \emph{{SIAM} Journal on Discrete Mathematics}, 38\penalty0 (1):\penalty0 1007--1029, 2024.

\bibitem[Anshelevich and Postl(2017)]{anshelevich2017randomized}
Elliot Anshelevich and John Postl.
\newblock Randomized social choice functions under metric preferences.
\newblock \emph{Journal of Artificial Intelligence Research}, 58:\penalty0 797--827, 2017.

\bibitem[Anshelevich and Sekar(2016)]{Anshelevich2016truthful}
Elliot Anshelevich and Shreyas Sekar.
\newblock Truthful mechanisms for matching and clustering in an ordinal world.
\newblock In \emph{Proceedings of the 12th International Conference Web and Internet Economics ({WINE})}, pages 265--278, 2016.

\bibitem[Anshelevich and Zhu(2019)]{Anshelevich2019bipartitematching}
Elliot Anshelevich and Wennan Zhu.
\newblock Tradeoffs between information and ordinal approximation for bipartite matching.
\newblock \emph{Theory of Computing Systems}, 63\penalty0 (7):\penalty0 1499--1530, 2019.

\bibitem[Anshelevich et~al.(2018)Anshelevich, Bhardwaj, Elkind, Postl, and Skowron]{anshelevich2018approximating}
Elliot Anshelevich, Onkar Bhardwaj, Edith Elkind, John Postl, and Piotr Skowron.
\newblock Approximating optimal social choice under metric preferences.
\newblock \emph{Artificial Intelligence}, 264:\penalty0 27--51, 2018.

\bibitem[Anshelevich et~al.(2021)Anshelevich, Filos-Ratsikas, Shah, and Voudouris]{distortion-survey}
Elliot Anshelevich, Aris Filos-Ratsikas, Nisarg Shah, and Alexandros~A. Voudouris.
\newblock Distortion in social choice problems: The first 15 years and beyond.
\newblock In \emph{Proceedings of the 30th International Joint Conference on Artificial Intelligence {(IJCAI)}}, pages 4294--4301, 2021.

\bibitem[Anshelevich et~al.(2025)Anshelevich, Filos-Ratsikas, Jerrett, and Voudouris]{anshelevich2025improved}
Elliot Anshelevich, Aris Filos-Ratsikas, Christopher Jerrett, and Alexandros~A Voudouris.
\newblock Improved metric distortion via threshold approvals.
\newblock \emph{Artificial Intelligence}, 341:\penalty0 104295, 2025.

\bibitem[Balkanski and Zhu(2025)]{BZ25}
Eric Balkanski and Cherlin Zhu.
\newblock Strategyproof learning with advice.
\newblock In \emph{Proceedings of the 36th International Conference on Algorithmic Learning Theory ({ALT})}, 2025.

\bibitem[Balkanski et~al.(2023)Balkanski, Gkatzelis, and Tan]{BGT23}
Eric Balkanski, Vasilis Gkatzelis, and Xizhi Tan.
\newblock Strategyproof scheduling with predictions.
\newblock In \emph{Proceedings of the 14th Innovations in Theoretical Computer Science Conference ({ITCS})}, pages 11:1--11:22, 2023.

\bibitem[Balkanski et~al.(2024{\natexlab{a}})Balkanski, Gkatzelis, and Shahkarami]{balkanski2024randomized}
Eric Balkanski, Vasilis Gkatzelis, and Golnoosh Shahkarami.
\newblock Randomized strategic facility location with predictions.
\newblock In \emph{Proceedings of The 38th Annual Conference on Neural Information Processing Systems ({NeurIPS})}, 2024{\natexlab{a}}.

\bibitem[Balkanski et~al.(2024{\natexlab{b}})Balkanski, Gkatzelis, Tan, and Zhu]{balkanski2023online}
Eric Balkanski, Vasilis Gkatzelis, Xizhi Tan, and Cherlin Zhu.
\newblock Online mechanism design with predictions.
\newblock In \emph{Proceedings of the 25th ACM Conference on Economics and Computation ({EC})}, page 1184, 2024{\natexlab{b}}.

\bibitem[Barak et~al.(2024)Barak, Gupta, and Talgam-Cohen]{BGI24}
Zohar Barak, Anupam Gupta, and Inbal Talgam-Cohen.
\newblock Mac advice for facility location mechanism design.
\newblock In \emph{Proceedings of the 38th Annual Conference on Neural Information Processing Systems ({NeurIPS})}, 2024.

\bibitem[Benade et~al.(2021)Benade, Nath, Procaccia, and Shah]{benade2021preference}
Gerdus Benade, Swaprava Nath, Ariel~D Procaccia, and Nisarg Shah.
\newblock Preference elicitation for participatory budgeting.
\newblock \emph{Management Science}, 67\penalty0 (5):\penalty0 2813--2827, 2021.

\bibitem[Berger et~al.(2024)Berger, Feldman, Gkatzelis, and Tan]{BFGT24}
Ben Berger, Michal Feldman, Vasilis Gkatzelis, and Xizhi Tan.
\newblock Learning-augmented metric distortion via (p, q)-veto core.
\newblock In \emph{Proceedings of the 25th {ACM} Conference on Economics and Computation ({EC})}, page 984, 2024.

\bibitem[Bhaskar et~al.(2018)Bhaskar, Dani, and Ghosh]{bhaskar2018truthful}
Umang Bhaskar, Varsha Dani, and Abheek Ghosh.
\newblock Truthful and near-optimal mechanisms for welfare maximization in multi-winner elections.
\newblock In \emph{Proceedings of the AAAI Conference on Artificial Intelligence {(AAAI)}}, pages 925--932, 2018.

\bibitem[Bogomolnaia and Moulin(2001)]{bogomolnaia2001new}
Anna Bogomolnaia and Herv{\'e} Moulin.
\newblock A new solution to the random assignment problem.
\newblock \emph{Journal of Economic theory}, 100\penalty0 (2):\penalty0 295--328, 2001.

\bibitem[Boutilier et~al.(2015)Boutilier, Caragiannis, Haber, Lu, Procaccia, and Sheffet]{boutilier2015optimal}
Craig Boutilier, Ioannis Caragiannis, Simi Haber, Tyler Lu, Ariel~D. Procaccia, and Or~Sheffet.
\newblock Optimal social choice functions: A utilitarian view.
\newblock \emph{Artificial Intelligence}, 227:\penalty0 190--213, 2015.

\bibitem[Brandt et~al.(2012)Brandt, Conitzer, and Endriss]{brandt2012computational}
Felix Brandt, Vincent Conitzer, and Ulle Endriss.
\newblock Computational social choice.
\newblock \emph{Multiagent systems}, 2:\penalty0 213--284, 2012.

\bibitem[Caragiannis and Fehrs(2024)]{caragiannis2023beyond}
Ioannis Caragiannis and Karl Fehrs.
\newblock Beyond the worst case: Distortion in impartial culture electorate.
\newblock In \emph{Proceedings of the 20th Conference on Web and Internet Economics ({WINE})}, 2024.

\bibitem[Caragiannis and Kalantzis(2024)]{caragiannis2024randomized}
Ioannis Caragiannis and Georgios Kalantzis.
\newblock Randomized learning-augmented auctions with revenue guarantees.
\newblock In \emph{Proceedings of the 33rd International Joint Conference on Artificial Intelligence ({IJCAI})}, pages 2687--2694, 2024.

\bibitem[Caragiannis and Procaccia(2011)]{caragiannis2011voting}
Ioannis Caragiannis and Ariel~D. Procaccia.
\newblock Voting almost maximizes social welfare despite limited communication.
\newblock \emph{Artificial Intelligence}, 175\penalty0 (9-10):\penalty0 1655--1671, 2011.

\bibitem[Caragiannis et~al.(2017)Caragiannis, Nath, Procaccia, and Shah]{caragiannis2017subset}
Ioannis Caragiannis, Swaprava Nath, Ariel~D. Procaccia, and Nisarg Shah.
\newblock Subset selection via implicit utilitarian voting.
\newblock \emph{Journal of Artificial Intelligence Research}, 58:\penalty0 123--152, 2017.

\bibitem[Caragiannis et~al.(2022)Caragiannis, Shah, and Voudouris]{caragiannis2022multiwinner}
Ioannis Caragiannis, Nisarg Shah, and Alexandros~A. Voudouris.
\newblock The metric distortion of multiwinner voting.
\newblock \emph{Artificial Intelligence}, 313:\penalty0 103802, 2022.

\bibitem[Caragiannis et~al.(2024)Caragiannis, Filos{-}Ratsikas, Frederiksen, Hansen, and Tan]{caragiannis2024augmentation}
Ioannis Caragiannis, Aris Filos{-}Ratsikas, S{\o}ren Kristoffer~Stiil Frederiksen, Kristoffer~Arnsfelt Hansen, and Zihan Tan.
\newblock Truthful facility assignment with resource augmentation: an exact analysis of serial dictatorship.
\newblock \emph{Mathematical Programming}, 203\penalty0 (1):\penalty0 901--930, 2024.

\bibitem[Charikar et~al.(2024)Charikar, Ramakrishnan, Wang, and Wu]{charikar24breaking}
Moses Charikar, Prasanna Ramakrishnan, Kangning Wang, and Hongxun Wu.
\newblock Breaking the metric voting distortion barrier.
\newblock In \emph{Proceedings of the 35th ACM-SIAM Symposium on Discrete Algorithms ({SODA})}, pages 1621--1640, 2024.

\bibitem[Chen et~al.(2024)Chen, Gravin, and Im]{QNS24}
Qingyun Chen, Nick Gravin, and Sungjin Im.
\newblock Strategic facility location via predictions.
\newblock In \emph{Proceedings of the 20th Conference on Web and Internet Economics ({WINE})}, 2024.

\bibitem[Christodoulou et~al.(2024)Christodoulou, Sgouritsa, and Vlachos]{christodoulou2024mechanism}
George Christodoulou, Alkmini Sgouritsa, and Ioannis Vlachos.
\newblock Mechanism design augmented with output advice.
\newblock In \emph{Proceedings of The 38th Annual Conference on Neural Information Processing Systems ({NeurIPS})}, 2024.

\bibitem[Colini-Baldeschi et~al.(2024)Colini-Baldeschi, Klumper, Sch{\"a}fer, and Tsikiridis]{colini2024trust}
Riccardo Colini-Baldeschi, Sophie Klumper, Guido Sch{\"a}fer, and Artem Tsikiridis.
\newblock To trust or not to trust: Assignment mechanisms with predictions in the private graph model.
\newblock In \emph{Proceedings of the 25th ACM Conference on Economics and Computation ({EC})}, pages 1134--1154, 2024.

\bibitem[Ebadian and Shah(2025)]{ES25}
Soroush Ebadian and Nisarg Shah.
\newblock Every bit helps: Achieving the optimal distortion with a few queries.
\newblock In \emph{Proceedings of 39th Annual AAAI Conference on Artificial Intelligence}, 2025.

\bibitem[Feldman et~al.(2016)Feldman, Fiat, and Golomb]{feldman2016voting}
Michal Feldman, Amos Fiat, and Iddan Golomb.
\newblock On voting and facility location.
\newblock In \emph{Proceedings of the 2016 {ACM} Conference on Economics and Computation ({EC})}, pages 269--286, 2016.

\bibitem[Filos-Ratsikas and Miltersen(2014)]{filos2014truthful}
Aris Filos-Ratsikas and Peter~Bro Miltersen.
\newblock Truthful approximations to range voting.
\newblock In \emph{Proceedings of the 10th International Conference on Web and Internet Economics ({WINE})}, pages 175--188, 2014.

\bibitem[Filos{-}Ratsikas et~al.(2014)Filos{-}Ratsikas, Frederiksen, and Zhang]{filos2014RSD}
Aris Filos{-}Ratsikas, S{\o}ren Kristoffer~Stiil Frederiksen, and Jie Zhang.
\newblock Social welfare in one-sided matchings: Random priority and beyond.
\newblock In \emph{Proceedings of the 7th International Symposium Algorithmic Game Theory ({SAGT})}, pages 1--12, 2014.

\bibitem[Gkatzelis et~al.(2020)Gkatzelis, Halpern, and Shah]{gkatzelis2020resolving}
Vasilis Gkatzelis, Daniel Halpern, and Nisarg Shah.
\newblock Resolving the optimal metric distortion conjecture.
\newblock In \emph{Proceedings of the 61st Annual IEEE Symposium on Foundations of Computer Science ({FOCS})}, pages 1427--1438, 2020.

\bibitem[Gkatzelis et~al.(2022)Gkatzelis, Kollias, Sgouritsa, and Tan]{gkatzelis2022improved}
Vasilis Gkatzelis, Kostas Kollias, Alkmini Sgouritsa, and Xizhi Tan.
\newblock Improved price of anarchy via predictions.
\newblock In \emph{Proceedings of the 23rd ACM Conference on Economics and Computation ({EC})}, pages 529--557, 2022.

\bibitem[Gkatzelis et~al.(2025)Gkatzelis, Schoepflin, and Tan]{gkatzelis2025clock}
Vasilis Gkatzelis, Daniel Schoepflin, and Xizhi Tan.
\newblock Clock auctions augmented with unreliable advice.
\newblock In \emph{Proceedings of the 2025 Annual ACM-SIAM Symposium on Discrete Algorithms (SODA)}, pages 2629--2655, 2025.

\bibitem[Hylland and Zeckhauser(1979)]{hylland1979efficient}
Aanund Hylland and Richard Zeckhauser.
\newblock The efficient allocation of individuals to positions.
\newblock \emph{Journal of Political economy}, 87\penalty0 (2):\penalty0 293--314, 1979.

\bibitem[Kizilkaya and Kempe(2022)]{kempe2022veto}
Fatih~Erdem Kizilkaya and David Kempe.
\newblock Plurality veto: {A} simple voting rule achieving optimal metric distortion.
\newblock In \emph{Proceedings of the 31st International Joint Conference on Artificial Intelligence ({IJCAI})}, pages 349--355, 2022.

\bibitem[Latifian and Voudouris(2024)]{latifian2024approval}
Mohamad Latifian and Alexandros~A. Voudouris.
\newblock The distortion of threshold approval matching.
\newblock In \emph{Proceedings of the 33rd International Joint Conference on Artificial Intelligence {(IJCAI)}}, pages 2851--2859, 2024.

\bibitem[Lu et~al.(2024)Lu, Wan, and Zhang]{lu2024competitive}
Pinyan Lu, Zongqi Wan, and Jialin Zhang.
\newblock Competitive auctions with imperfect predictions.
\newblock In \emph{Proceedings of the 25th ACM Conference on Economics and Computation ({EC})}, pages 1155--1183, 2024.

\bibitem[Lykouris and Vassilvitskii(2021)]{lykouris2021competitive}
Thodoris Lykouris and Sergei Vassilvitskii.
\newblock Competitive caching with machine learned advice.
\newblock \emph{Journal of the ACM}, 68\penalty0 (4):\penalty0 1--25, 2021.

\bibitem[Ma et~al.(2021)Ma, Menon, and Larson]{ma2021matching}
Thomas Ma, Vijay Menon, and Kate Larson.
\newblock Improving welfare in one-sided matchings using simple threshold queries.
\newblock In \emph{Proceedings of the 30th International Joint Conference on Artificial Intelligence {(IJCAI)}}, pages 321--327, 2021.

\bibitem[Mitzenmacher and Vassilvitskii(2021)]{mitzenmacher2022algorithms}
Michael Mitzenmacher and Sergei Vassilvitskii.
\newblock \emph{Algorithms with Predictions}, pages 646--662.
\newblock Cambridge University Press, 2021.

\bibitem[Procaccia and Rosenschein(2006)]{procaccia2006distortion}
Ariel~D. Procaccia and Jeffrey~S. Rosenschein.
\newblock The distortion of cardinal preferences in voting.
\newblock In \emph{Proceedings of the 10th International Workshop on Cooperative Information Agents ({CIA})}, pages 317--331, 2006.

\bibitem[Von~Neumann and Morgenstern(1944)]{von1944theory}
J.~Von~Neumann and O.~Morgenstern.
\newblock Theory of games and economic behavior.
\newblock 1944.

\bibitem[Xu and Lu(2022)]{XL22}
Chenyang Xu and Pinyan Lu.
\newblock Mechanism design with predictions.
\newblock In \emph{Proceedings of the 31st International Joint Conference on Artificial Intelligence ({IJCAI})}, pages 571--577, 2022.

\end{thebibliography}

\end{document}